\documentclass[12pt, twoside]{article}

\usepackage{amsmath,amsthm,amssymb}

\usepackage{times}

\usepackage{enumerate}

 \usepackage{amsmath}
\usepackage{amssymb}
\usepackage{amsthm}
\usepackage{graphicx}
\usepackage{amscd}
\usepackage{graphicx}

\usepackage{amssymb}
\usepackage{amsmath}
\usepackage{latexsym}
\usepackage[numbers,sort&compress]{natbib}
\usepackage{verbatim}
\usepackage{physics}
\usepackage{dsfont}	
\usepackage{microtype}  
\usepackage{hyperref}
\usepackage{tikz}
\usepackage{booktabs}   

\usepackage{color}

\newtheorem{theorem}{Theorem}
\newtheorem{lemma}[theorem]{Lemma}

\begin{document}

 \title{The Grassl-R\"otteler cyclic and consta-cyclic MDS codes are generalised Reed-Solomon codes}
\author{Simeon Ball\thanks{15 June 2022. The author acknowledges the support of MTM2017-82166-P and PID2020-113082GB-I00 financed by MCIN / AEI / 10.13039/501100011033, the Spanish Ministry of Science and Innovation.}}
 \date{}
\maketitle

\begin{abstract}
We prove that the cyclic and constacyclic codes constructed by Grassl and R\"otteler in \cite{GR2015} are generalised Reed-Solomon codes. 
This note can be considered as an addendum to Grassl and R\"otteler \cite{GR2015}. It can also be considered as an appendix to Ball and Vilar \cite{BV2021}, where Conjecture 11 of \cite{GR2015}, which was stated for Grassl-R\"otteler codes, is proven for generalised Reed-Solomon codes. The content of this note, together with \cite{BV2021}, therefore implies that Conjecture 11 from \cite{GR2015} is true.
\end{abstract}

\section{Introduction}

Let ${\mathbb F}_q$ denote the finite field with $q$ elements. 

The weight of an element of ${\mathbb F}_q^n$ is the number of non-zero coordinates that it has.

A $k$-dimensional linear code of length $n$ and minimum distance $d$ over ${\mathbb F}_q$, denoted as a $[n,k,d]_q$ code, is a $k$-dimensional subspace of ${\mathbb F}_q^n$ in which every non-zero vector has weight at least $d$.

The Singleton bound for linear codes states that
$$
n \geqslant k+d-1
$$
 and a linear code which attains the Singleton bound is called a maximum distance separable codes, or MDS code for short.
 
It is a simple matter to prove the bound $n \leqslant q+k-1$ and the MDS conjecture, for linear codes, states that
if $4\leqslant k \leqslant q-2$ then
$$
n \leqslant q+1.
$$
For values of $k$ outside of this range it is not difficult to determine the longest length of a linear MDS code. The MDS conjecture is known to hold for $q$ prime \cite{Ball2012}, where it was also proven that if $k \neq (q+1)/2$ and $q$ is prime then a $[q+1,k,q+2-k]_q$ MDS code is a generalised Reed-Solomon code.

Let $\{a_1,\ldots,a_q\}$ be the set of elements of ${\mathbb F}_q$.

 A {\em generalised Reed-Solomon code} over ${\mathbb F}_{q}$ is
\begin{equation} \label{GRScodede}
D=\{(\theta_1f(a_1),\ldots,\theta_{q}f(a_{q}),\theta_{q+1}f_{k-1}) \ | \ f \in {\mathbb F}_{q}[X],\ \deg f \leqslant k-1\},
\end{equation}
where $f_i$ denotes the coefficient of $X^i$ in $f(X)$ and $\theta_i \in {\mathbb F}_{q} \setminus \{0 \}$. 

The Reed-Solomon code is the case in which $\theta_j=1$, for all $j$.

We note that our definition of a (generalised) Reed-Solomon code is what some authors call the extended or doubly extended Reed-Solomon code. That is, many authors do not include the final coordinate or the evaluation at zero. However, a more natural definition of the Reed-Solomon code, which is entirely equivalent to the above, is obtained by evaluating homogeneous polynomials $f \in {\mathbb F}_{q}[X_1,X_2]$ of degree $k-1$, at the points of the projective line,
\begin{equation} \label{GRScodede2}
D=\{(\theta_1f(a_1,1),\ldots,\theta_{q}f(a_{q},1),\theta_{q+1}f(1,0)) \ | \ f \in {\mathbb F}_{q}[X_1,X_2],\ \ f \ \mathrm{homogeneous}, \ \ \deg f = k-1\}.
\end{equation}

\section{Generalised Reed-Solomon codes}

In this section we prove that a generalised Reed-Solomon code can be constructed as an evaluation code, evaluating at the $(q+1)$-st roots of unity of ${\mathbb F}_{q^2}$. Thus, any generalised Reed-Solomon code can be obtained in this way by multiplying the $i$-th coordinate by a non-zero $\theta_i \in {\mathbb F}_q$, as in definition (\ref{GRScodede}) and (\ref{GRScodede2}).

Let $\{\alpha_1,\ldots,\alpha_{q+1}\}$ be the set of $(q+1)$-st roots of unity of ${\mathbb F}_{q^2}$.

\begin{lemma} \label{DERSiscyclic}
If $k$ is odd then the code
$$
C=\{(h(\alpha_1)+h(\alpha_1)^q,\ldots,h(\alpha_{q+1})+h(\alpha_{q+1})^{q}) \ | \ h \in {\mathbb F}_{q^2}[X],\ \deg h \leqslant \tfrac{1}{2}(k-1) \}
$$
is a $[q+1,k,q+2-k]_q$ generalised Reed-Solomon code.
\end{lemma} 

\begin{proof}
Note that $C$ is a subspace over ${\mathbb F}_q$ and that it has size $q^k$ since the constant term of $$
h(X)+h(X)^q
$$
is an element of ${\mathbb F}_q$. Thus, $C$ is a $k$-dimensional subspace of ${\mathbb F}_q^{q+1}$.

Let
$$
h(X)=\sum_{i=0}^{(k-1)/2} c_iX^i.
$$
Suppose that $\{1,e\}$ is a basis for ${\mathbb F}_{q^2}$ over ${\mathbb F}_q$. 

For $\alpha$, a $(q+1)$-st root of unity, let $x_1, x_2 \in {\mathbb F}_q$ be such that
$$
\alpha=(x_1+ex_2)^{q-1}.
$$
Observe that as $(x_1,x_2)$ vary over the points of the projective line, $\alpha$ will run through the distinct $(q+1)$-st roots of unity.

Then
$$
h(\alpha)+h(\alpha)^q=\sum_{i=0}^{(k-1)/2} c_i (x_1+ex_2)^{i(q-1)}+c_i^q (x_1+ex_2)^{i(1-q)}
$$
$$
=\sum_{i=0}^{(k-1)/2} c_i (x_1+e^qx_2)^{i}(x_1+ex_2)^{-i}+c_i^q (x_1+ex_2)^{i}(x_1+e^qx_2)^{-i}
$$
$$
=(x_1+ex_2)^{-(k-1)(q+1)/2}\Big(\sum_i c_i (x_1+ex_2)^{(k-1)/2-i}(x_1+e^qx_2)^{(k-1)/2+i} 
$$
$$
+c_i^q (x_1+ex_2)^{(k-1)/2+i}(x_1+e^qx_2)^{(k-1)/2-i}\Big).
$$
Note that $(x_1+ex_2)^{-(k-1)(q+1)/2} \in {\mathbb F}_q$, does not depend on $h(X)$. 

Thus, the coefficient of $x_1^jx_2^{k-j-1}$ of 
$$
\sum_{i=0}^{(k-1)/2} c_i (x_1+ex_2)^{(k-1)/2-i}(x_1+e^qx_2)^{(k-1)/2+i}+c_i^q (x_1+ex_2)^{(k-1)/2+i}(x_1+e^qx_2)^{(k-1)/2-i},
$$
is also an element of ${\mathbb F}_q$. Hence, the $\alpha$ coordinate of a codeword of $C$ is the evaluation of a homogeneous polynomial in ${\mathbb F}_q[x_1,x_2]$ of degree $k-1$, multiplied by a non-zero element of ${\mathbb F}_q$. By definition (\ref{GRScodede2}), we conclude that such a code $C$ is a generalised Reed-Solomon code.
\end{proof}

The previous lemma only applies to the case when $k$ is odd. The following lemma deals with the case $k$ is even.

\begin{lemma} \label{DERSiscyclic2}
For $\alpha_i$, a $(q+1)$-st root of unity, let $\omega_i$ be such that $\alpha_i=\omega_i^{q-1}$. If $k$ is even then the code
$$
C=\{ \omega_1^qh(\alpha_1)+\omega_1 h(\alpha_1)^q,\ldots, \omega_{q+1}^{q} h(\alpha_{q+1})+\omega_{q+1}h(\alpha_{q+1})^{q}) \ | \ h \in {\mathbb F}_{q^2}[X],\ \deg h \leqslant \tfrac{1}{2}k-1 \}
$$
is a $[q+1,k,q+2-k]_q$ generalised Reed-Solomon code.
\end{lemma} 

\begin{proof}
The proof is similar to that of Lemma~\ref{DERSiscyclic}. In this case we have that, $\omega=x_1+ex_2$ and so
$$
\omega^qh(\alpha)+\omega h(\alpha)^q=\sum_{i=0}^{ \tfrac{1}{2}k-1} c_i (x_1+ex_2)^{i(q-1)+q}+c_i^q (x_1+ex_2)^{i(1-q)+1}
$$
$$
=(x_1+ex_2)^{-(\tfrac{1}{2}k-1)(q+1)}\left(\sum_i c_i (x_1+ex_2)^{\tfrac{1}{2}k-1-i}(x_1+e^qx_2)^{\tfrac{1}{2}k+i} \right.
$$
$$
\left.
+c_i^q (x_1+ex_2)^{\tfrac{1}{2}k+i}(x_1+e^qx_2)^{\tfrac{1}{2}k-1-i}\right).
$$
The coefficient of $x_1^jx_2^{k-j-1}$,
$$
\sum_i c_i (x_1+ex_2)^{\tfrac{1}{2}k-1-i}(x_1+e^qx_2)^{\tfrac{1}{2}k+i} +c_i^q (x_1+ex_2)^{\tfrac{1}{2}k+i}(x_1+e^qx_2)^{\tfrac{1}{2}k-1-i},
$$
is an element of ${\mathbb F}_q$. Thus, the lemma follows in the same way as Lemma~\ref{DERSiscyclic}.\end{proof}

\section{Grassl-R\"otteler cyclic and constacyclic MDS codes}

The $k$-dimensional cyclic or constacyclic code $\langle g \rangle$ of length $n$ over ${\mathbb F}_q$,
where
$$
g(X)=\sum_{j=0}^{n-k} c_j X^j \in {\mathbb F}_q[X],
$$
is a linear code of length $n$ spanned by the $k$ cyclic shifts of the codeword 
$$
(c_0,\ldots,c_{n-k},0,\ldots,0).
$$
It is a cyclic code if $g$ divides $X^{n}-1$ and constacyclic code if $g$ divides $X^{n}-\eta$, for some $\eta \neq 1$.
See \cite{Ball2020} or \cite{VanLint1999} for the basic results concerning cyclic codes.

In \cite{GR2015}, Grassl and R\"otteler introduced three $[q+1,k,q+2-k]_q$ MDS codes, the first two are constructed as cyclic codes and the third as a constacyclic code. As mentioned in the introduction, it follows from \cite{Ball2012} that when $q$ is prime, these codes are generalised Reed-Solomon codes. In this section we shall prove that they are generalised Reed-Solomon codes for all $q$.

Let $\omega$ be a primitive element of ${\mathbb F}_{q^2}$ and let $\alpha=w^{q-1}$, a primitive $(q+1)$-st root of unity.

The Grassl-R\"otteler codes depend on the parity of $q$ and $k$.

For $q$ and $k$ both odd, and $q$ and $k$ both even, the Grassl-R\"otteler code is $\langle g_1 \rangle$, where
$$
g_1(X)=\prod_{i=-r}^r (X- \alpha^i).
$$
For $k$ odd and $q$ even, the Grassl-R\"otteler code is the cyclic code $\langle g_2 \rangle$, where
$$
g_2(X)=\prod_{i=\frac{1}{2}q-r}^{\frac{1}{2}q+r+1} (X- \alpha^i).
$$
And for $k$ even and $q$ odd, the Grassl-R\"otteler code is the constacyclic code $\langle g_3 \rangle$, where
$$
g_3(X)=\prod_{i=-r+1}^{r} (X-\omega \alpha^i).
$$

It is a simple matter to check that for $i \in \{1,2,3\}$, $g_i \in {\mathbb F}_q[X]$ and for $i \in \{1,2\}$, the polynomial $g_i$ divides $X^{q+1}-1$ and $g_3$ divides $X^{q+1}-\omega^{q+1}$.

We now treat each of the four cases, which depends on the parity of $k$ and $q$, in turn and prove that they are all generalised Reed-Solomon codes.

Let $\{e_1,\ldots,e_{q+1}\}$ be the canonical basis of ${\mathbb F}_q^{q+1}$.

Let $\beta \in {\mathbb F}_{q^2}$ be such that $\beta+\beta^q=1$.

\begin{theorem} \label{koddqodd}
If $k$ and $q$ are both odd then the $[q+1,k,q+2-k]_q$ code $\langle g_1 \rangle$ is a generalised Reed-Solomon code.
\end{theorem}

\begin{proof}
Let $c_j$ be defined by
$$
g_1(X)=\prod_{i=-r}^r (X- \alpha^i)=\sum_{j=0}^{2r+1} c_j X^j.
$$
Observe that $k=q-2r$.

We will prove that, for $a \in \{0,\ldots,k-1\}$,
$$
\sum_{s=a}^{q+1-k+a} (-1)^s c_{s-a} e_{s+1}
=(\underbrace{0,\ldots,0}_{a},(-1)^a c_0,\ldots,(-1)^{q+1-k+a}c_{q+1-k},\underbrace{0,\ldots,0}_{k-1-a})
$$
are the evaluations of polynomials,
$$
h(X)+h(X)^q
$$
where $h \in {\mathbb F}_{q^2}[X]$ is of degree at most $(k-1)/2$, evaluated at the $(q+1)$-st roots of unity. 

Lemma~\ref{DERSiscyclic} implies that if we multiply the $(s+1)$-th coordinate of the codewords in $\langle g_1\rangle$ by $(-1)^s$ then we obtain a generalised Reed-Solomon code, which implies that $\langle g_1\rangle$ is a generalised Reed-Solomon code.

For $a \in \{0,\ldots,k-1\}$, define
$$
h_a(X)=\sum_{i=1}^{(q-1)/2} \sum_{j=0}^{2r+1} c_j \alpha^{i(j+a)} X^{(q+1)/2-i}+ \sum_{j=0}^{2r+1} c_j(-1)^{j+a}\beta+\sum_{j=0}^{2r+1} c_j \beta X^{(q+1)/2}.
$$
For all $i \in \{0,\ldots,r\}$,
$$
\sum_{j=0}^{q} c_j \alpha^{ij}=0,
$$
since $g_1(\alpha^i)=0$.
 Thus, the degree of $h_a$ is at most $(q-1)/2-r=(k-1)/2$. 

We have that
$$
h_a(\alpha^s)=\sum_{i=1}^{(q-1)/2} \sum_{j=0}^{2r+1} c_j \alpha^{i(j+a-s)}(-1)^s+ \sum_{j=0}^{2r+1} c_j(-1)^{j+a}\beta+\sum_{j=0}^{2r+1} c_j \beta(-1)^s.
$$
Since,
$$
(\sum_{i=1}^{(q-1)/2} c_j \alpha^{i(j+a-s)})^q=\sum_{i=(q+3)/2}^{q}c_j   \alpha^{i(j+a-s)},
$$
and $\beta+\beta^q=1$, it follows that
$$
h_a(\alpha^s)+h_a(\alpha^s)^q=(-1)^s  \sum_{j=0}^{2r+1}\sum_{i=0}^{q}c_j\alpha^{i(j+a-s)}.
$$
Since $\sum_{i=0}^{q} \alpha^{ij}=0$ unless $j=0$, in which case it is one,
$$
h_a(\alpha^s)+h_a(\alpha^s)^q=(-1)^sc_{s-a},
$$
which is precisely what we had to prove.
\end{proof}

We next deal with the case $k$ and $q$ are both even, since this is again the code $\langle g_1 \rangle$.

\begin{theorem} \label{kevenqeven}
If $k$ and $q$ are both even then the $[q+1,k,q+2-k]_q$ code $\langle g_1 \rangle$ is a generalised Reed-Solomon code.
\end{theorem}

\begin{proof}
We can simply copy the proof of Theorem~\ref{koddqodd} until we define $h_a(X)$. Then we have to define $h_a(X)$ differently, partly because we will apply Lemma~\ref{DERSiscyclic2} in place of Lemma~\ref{DERSiscyclic}.

For $a \in \{0,\ldots,k-1\}$, define
$$
h_a(X)=\sum_{i=1}^{\tfrac{1}{2}q} \sum_{j=0}^{2r+1} c_j \alpha^{i(j+a)} X^{\tfrac{1}{2}q-i}+\sum_{j=0}^{2r+1} c_j \beta X^{\tfrac{1}{2}q}.
$$
Observe that, since $g_1(\alpha^i)=0$, which implies that
$$
\sum_{j=0}^{q} c_j \alpha^{ij}=0
$$
for all $i \in \{0,\ldots,r\}$. Thus, the degree of $h_a$ is at most $\tfrac{1}{2}q-r-1=\tfrac{1}{2}k-1$. 

As before, let $\omega$ be a fixed primitive element of ${\mathbb F}_{q^2}$ and let $\alpha=\omega^{q-1}$, a primitive $(q+1)$-st root of unity. Then
$$
h_a(\alpha^s)=\sum_{i=1}^{\tfrac{1}{2}q} \sum_{j=0}^{2r+1} c_j \alpha^{i(j+a-s)}\alpha^{\tfrac{1}{2}sq}+ \sum_{j=0}^{2r+1} c_j\beta \alpha^{\tfrac{1}{2}sq}.
$$
and so
$$
\alpha^{-s} h_a(\alpha^s)^q=\sum_{i=1}^{\tfrac{1}{2}q} \sum_{j=0}^{2r+1} c_j \alpha^{-i(j+a-s)}\alpha^{-\tfrac{1}{2}sq-s}+ \sum_{j=0}^{2r+1} c_j\beta^q \alpha^{-\tfrac{1}{2}sq-s}.
$$

Since, $\beta+\beta^q=1$ and $\alpha^{-\tfrac{1}{2}sq-s}=\alpha^{\tfrac{1}{2}sq}$, it follows that
$$
h_a(\alpha^s)+\alpha^{-s}h_a(\alpha^s)^q=\alpha^{\tfrac{1}{2}sq} \sum_{j=0}^{2r+1}\sum_{i=0}^{q}c_j\alpha^{i(j+a-s)}.
$$
Since $\sum_{i=0}^{q} \alpha^{ij}=0$ unless $j=0$, in which case it is one,
$$
h_a(\alpha^s)+\alpha^{-s}h_a(\alpha^s)^q= \alpha^{\tfrac{1}{2}sq} c_{s-a}.
$$
Hence,
$$
\omega^{sq} h_a(\alpha^s)+\omega^{s}h_a(\alpha^s)^q=\omega^{\tfrac{1}{2}s(q+1)} c_{s-a}.
$$

Lemma~\ref{DERSiscyclic2} implies that if we multiply the $(s+1)$-th coordinate of the codewords in $\langle g_1\rangle$ by $\omega^{\tfrac{1}{2}s(q+1)}$ then we obtain a generalised Reed-Solomon code, which implies that $\langle g_1\rangle$ is a generalised Reed-Solomon code.
\end{proof}

The next theorem deals with the case $k$ is odd and $q$ is even. In this case the Grassl-R\"otteler code is  $\langle g_2 \rangle$.

\begin{theorem} \label{koddqeven}
If $k$ is odd and $q$ is even then the $[q+1,k,q+2-k]_q$ code $\langle g_2 \rangle$ is a generalised Reed-Solomon code.
\end{theorem}

\begin{proof}
Let $c_j$ be defined by
$$
g_2(X)=\prod_{i=\tfrac{1}{2}q-r}^{\tfrac{1}{2}q+r+1} (X- \alpha^i)=\sum_{j=0}^{2r+2} c_j X^j.
$$
Observe that $k=q-2r-1$.

As in Theorem~\ref{koddqodd}, we look for polynomials $h_a(X)$ which allow us to apply Lemma~\ref{DERSiscyclic}.

For $a \in \{0,\ldots,k-1\}$, let
$$
h_a(X)=\sum_{i=1}^{\tfrac{1}{2}q} \sum_{j=0}^{2r+2} c_j \alpha^{(i+\frac{1}{2}q)(j+a)} X^{\tfrac{1}{2}q+1-i}+ \sum_{j=0}^{2r+2} c_j \beta.
$$
Observe that, for all $i \in \{\tfrac{1}{2}q+1,\ldots,\tfrac{1}{2}q+r+1\}$, 
$$
\sum_{j=0}^{q} c_j \alpha^{ij}=0,
$$
since $g_1(\alpha^i)=0$. Thus, the degree of $h_a$ is at most $\tfrac{1}{2}q+1-(r+2)=\frac{1}{2}(k-1)$. 

We have that
$$
h_a(\alpha^s)=\sum_{i=1}^{\tfrac{1}{2}q} \sum_{j=0}^{2r+2} c_j \alpha^{(i+\tfrac{1}{2}q)(j+a-s)}+ \sum_{j=0}^{2r+2} c_j \beta.
$$
and so
$$
h_a(\alpha^s)^q=\sum_{i=1}^{\tfrac{1}{2}q} \sum_{j=0}^{2r+2} c_j \alpha^{(-i+\tfrac{1}{2}q+1)(j+a-s)}+ \sum_{j=0}^{2r+2} c_j \beta^q.
$$

Since, $\beta+\beta^q=1$, it follows that
$$
h_a(\alpha^s)+h_a(\alpha^s)^q=\sum_{j=0}^{2r+2}\sum_{i=0}^{q}c_j\alpha^{i(j+a-s)}.
$$
Since $\sum_{i=0}^{q} \alpha^{ij}=0$ unless $j=0$, in which case it is one,
$$
h_a(\alpha^s)+h_a(\alpha^s)^q=  c_{s-a}.
$$
Lemma~\ref{DERSiscyclic2} implies that $\langle g_1\rangle$  is a generalised Reed-Solomon code.
\end{proof}

Finally, we deal with the case $k$ is even and $q$ is odd, which is the constacyclic code $\langle g_3 \rangle$.

\begin{theorem} \label{kevenqodd}
If $k$ is even and $q$ is odd then the $[q+1,k,q+2-k]_q$ code $\langle g_3 \rangle$ is a generalised Reed-Solomon code.
\end{theorem}

\begin{proof}
Let $c_j$ be defined by
$$
g_3(X)=\prod_{i=-r+1}^{r} (X- \omega \alpha^i)=\sum_{j=0}^{2r} c_j X^j.
$$
Observe that $k=q-2r+1$.

As in Theorem~\ref{kevenqeven}, we look for polynomials $h_a(X)$ which allow us to apply Lemma~\ref{DERSiscyclic2}.

For $a \in \{0,\ldots,k-1\}$, let
$$
h_a(X)=\sum_{i=1}^{\tfrac{1}{2}(q+1)} \sum_{j=0}^{2r} \omega^{j+a} c_j \alpha^{i(j+a)} X^{\tfrac{1}{2}(q+1)-i}.
$$
Observe that, for all $i \in \{0,\ldots,r\}$, 
$$
\sum_{j=0}^{2r} c_j \omega^j \alpha^{ij}=0,
$$
since $g_3(\omega \alpha^i)=0$. Thus, the degree of $h_a$ is at most $\tfrac{1}{2}(q+1)-(r+1)=\frac{1}{2}k-1$. 

We have that
$$
h_a(\alpha^s)=\sum_{i=1}^{\tfrac{1}{2}(q+1)} \sum_{j=0}^{2r}  \omega^{j+a} c_j \alpha^{i(j+a-s)}(-1)^s.
$$
and, since $\omega^q=\omega \alpha$,
$$
\alpha^{-s}h_a(\alpha^s)^q=\sum_{i=1}^{\tfrac{1}{2}(q+1)} \sum_{j=0}^{2r}  \omega^{j+a}  c_j \alpha^{-(i-1)(j+a-s)}(-1)^s.
$$

Hence, it follows that
$$
h_a(\alpha^s)+\alpha^{-s}h_a(\alpha^s)^q=\sum_{i=1}^{q+1} \sum_{j=0}^{2r} \omega^{j+a} c_j \alpha^{i(j+a-s)}(-1)^s.
$$
Since $\sum_{i=1}^{q+1} \alpha^{ij}=0$ unless $j=0$, in which case it is one,
$$
h_a(\alpha^s)+\alpha^{-s}h_a(\alpha^s)^q= \omega^{s}(-1)^s c_{s-a}.
$$
Hence,
$$
\omega^{sq}h_a(\alpha^s)+\omega^{s}h_a(\alpha^s)^q= \omega^{s(q+1)}(-1)^s c_{s-a}.
$$
Lemma~\ref{DERSiscyclic2} implies that if we multiply the $(s+1)$-th coordinate of the codewords in $\langle g_3\rangle$ by $(-w^{(q+1)})^s$ then we obtain a generalised Reed-Solomon code, which implies that $\langle g_3\rangle$ is a generalised Reed-Solomon code.
\end{proof}

\section{Conclusions}

This note was motivated by Conjecture 11 from \cite{GR2015} which states that the minimum distance $d$ of the puncture code of the Grassl-R\"otteler code satisfies
$$
d= \left\{ \begin{array}{ll} 
2k & \mathrm{if}\ 1 \leqslant k \leqslant q/2 \\ 
(q+1)(k-(q-1)/2) & \mathrm{if} \ (q+1)/2 \leqslant k \leqslant q-1, \ q \ \mathrm{odd}  \\
q(k+1-q/2) & \mathrm{if} \ q/2 \leqslant k \leqslant q-1, \ q \ \mathrm{even}  \\
q^2+1 & \mathrm{if} \ k=q.
\end{array} \right.
$$
This conjecture is proven in \cite{BV2021} for generalised Reed-Solomon codes, which combined with the content of this note, implies that Conjecture 11 from \cite{GR2015} is indeed true.

It may be an interesting and worthwhile exercise to see if the other known $[q+1,k,q+2-k]_q$ MDS codes can be easily obtained as evaluation codes, evaluating at the $(q+1)$-st roots of unity. It may even be that the evaluation is over a more exotic set of elements in some extension of ${\mathbb F}_q$.
For completeness sake, we mention the other known $[q+1,k,q+2-k]_q$ MDS codes.

For $k=3$ and $q$ even, there are many examples known. These can all be extended to a $[q+2,k,q+3-k]_q$ MDS code. The columns of a generator matrix of such a code can be viewed as a set of points in the projective plane PG$(2,q)$. Such a set of points is known as a {\em hyperoval}. For a complete list of known hyperovals, see~\cite[Table 1]{BL2019}.

There are only two other known examples, up to duality.

The following is due to Segre \cite{Segre1959}. The linear code whose columns are the elements of the set
$$
\{ (1,t,t^{2^e},t^{2^e+1}) \ | \ t \in {\mathbb F}_q \} \cup \{ (0,0,0,1)\}
$$
is a $[q+1,4,q-2]_q$ linear MDS code, whenever $q=2^h$ and $(e,h)=1$.

The other is due to Glynn \cite{Glynn1986}. Let $\eta$ be an element of ${\mathbb F}_9$ such that $\eta^4=-1$. The linear code whose columns are the elements of the set
$$
\{ (1,t,t^{2}+\eta t^6,t^{3},t^4) \ | \ t \in {\mathbb F}_9\} \cup \{ (0,0,0,0,1)\}.
$$
is a $[10,5,6]_9$ linear MDS code,

\vspace{1cm}

   Simeon Ball\\
   Departament de Matem\`atiques, \\
Universitat Polit\`ecnica de Catalunya, \\
M\`odul C3, Campus Nord,\\
Carrer Jordi Girona 1-3,\\
08034 Barcelona, Spain \\
   {\tt simeon.michael.ball@upc.edu} \\

 \end{document}